\documentclass[11pt]{article}

\usepackage {a4}

\usepackage [utf8]{inputenc}
\usepackage [english]{babel}
\usepackage{enumerate}
\usepackage{fullpage}
\usepackage{bbm}
\usepackage{braket}

\usepackage{graphicx}
\usepackage{amsmath,amsthm,amssymb,amsfonts,setspace}
\usepackage{bm}
\usepackage{mathtools}
\usepackage{braket}
\usepackage{float}
\usepackage{color}
\usepackage[usenames,dvipsnames]{xcolor}
\usepackage{hyperref}  
\usepackage[capitalize,nameinlink,noabbrev]{cleveref}
\usepackage{wrapfig}
\usepackage{tikz}
\usepackage[left=1in,right=1in,top=1in,bottom=1in]{geometry}

\hypersetup{colorlinks=true,urlcolor=[rgb]{0,0,0.5},citecolor=[rgb]{0.5,0,0},linkcolor=[rgb]{0,0,0.4}}

\DeclareUnicodeCharacter{00A0}{ }

\newcommand{\id}{\mathbbm{1}}

\newcommand{\norm}[1]{\left\Vert #1 \right\Vert}

\newcommand{\C}{\mathbb{C}}
\newcommand{\Pp}{\mathbb{P}}
\newcommand{\E}{\mathbb{E}}

\newcommand{\FF}{\mathcal{F}}

\newcommand{\supp}{\mathrm{supp}}

\newcommand{\phinet}{\Phi_{\mathrm{net}}}

\newcommand{\channels}{\mathrm{C}}
\newcommand{\operators}{\mathrm{B}}
\newcommand{\erasure}{\mathcal{E}}

\newcommand{\tr}{\mathrm{tr}}
\newcommand{\ketbra}[2]{\ket{#1}\!\!\bra{#2}}

\newcommand{\Cov}{\mathrm{Cov}}
\newcommand{\junk}{\mathrm{junk}}
\newcommand{\Span}{\mathrm{span}}
\newcommand{\intervalint}[1]{{{\left[\hspace{-1.4pt}\left[#1\right]\hspace{-1.4pt}\right]}}} 

\newtheorem{theorem}{Theorem}[section]
\newtheorem{definition}{Definition}[section]
\newtheorem{lemma}{Lemma}[section]
\newtheorem{proposition}{Proposition}[section]
\newtheorem{remark}{Remark}[section]
\newtheorem{problem}{Problem}[section]
\newtheorem{question}{Question}[section]
\newtheorem{corollary}{Corollary}[section]

\theoremstyle{definition}
\newtheorem{result}{Result}
\newtheorem{example}{Example}[section]

\usepackage[utf8]{inputenc}
\usepackage{todonotes}

\usepackage{authblk}

\begin{document}

\title{Randomized simulation of quantum channels using small ancilla} 

\author[1]{Marcin Kotowski \thanks{\!mkotowski@cft.edu.pl}}
\author[2]{Michał Kotowski \thanks{\!michal.kotowski@mimuw.edu.pl}}
\affil[1]{Center for Quantum Enabled-Computing, Center for Theoretical Physics of the
Polish Academy of Sciences, Al. Lotników 32/46, 02-668 Warsaw, Poland}
\affil[2]{University of Warsaw, Faculty of Mathematics, Informatics, and Mechanics, Institute of Mathematics, Banacha 2, 02-097 Warsaw, Poland}
\date{}                     
\setcounter{Maxaffil}{0}
\renewcommand\Affilfont{\itshape\small}
\date{}            
\maketitle

\begin{abstract}
We study the problem of implementing a quantum channel using a small ancilla with classical randomization and postselection on an output failure flag. The simulation is probabilistic, but exact conditioned on success. We prove that any unital channel on a $d$-dimensional system can be simulated with ancilla dimension $k$ and success probability $\Omega(\frac{k}{\log d})$. Equivalently, every unital channel on $n$ qubits can be simulated with constant success probability using only $O(\log n)$ ancillary qubits. We show that this tradeoff is best possible by constructing a family of channels which cannot be simulated with success probability better than $O(\frac{k}{\log d})$. We also show that the class of \textit{highly noncommutative} channels, which includes random channels, admits constant-success simulation with just a single ancillary qubit. We further show that this model of simulation necessarily fails for strongly non-unital channels and discuss possible extensions involving adaptivity. On the technical side we rely on a partition-based protocol and matrix concentration inequalities, including the recent refinement of noncommutative Khintchine inequalities due to  Bandeira, Boedihardjo and van Handel. 
\end{abstract}

\tableofcontents

\section{Introduction}

Quantum systems fundamentally evolve under unitary operators and such operations also form basic computational blocks of circuit-based quantum computing. However, in open systems the most general form of quantum maps is given by quantum channels (CPTP maps) and in fact many natural quantum information processing tasks can be cast as implementing a specific non-unitary quantum channel. Examples include measurement with discarded outcomes (the L\"uders channel), recovery map in quantum error correction \cite{beny}, simulating open system evolution by implementing Lindbladian evolutions $\Phi=e^{t\mathcal{L}}$ \cite{cleve} and dissipative state preparation that uses channels with prescribed fixed points \cite{verstraete, kraus}. The link between quantum channels and unitary operators is given by the classic Stinespring theorem \cite{watrous}, which states that any quantum channel can be implemented by performing a unitary on a larger system and then tracing out the environment.

A general quantum channel on a $d$-dimensional system may require as much as $d^2$-dimensional ancilla to implement using Stinespring theorem. For an $n$ qubit system the ancilla required is thus $2n$ additional qubits. Since adding ancillary qubits to a system is costly, it is natural to try to reduce the ancilla dimension as much as possible. One way to do this is by classical randomization, i.e. by implementing a random unitary operation. Such schemes are ubiquitous in quantum computing, e.g. in the form of picking random Pauli operators or Haar random unitaries. In this way one can implement the class of \textit{mixed unitary channels} with no ancilla. However, even restricting to unital channels, it is known that unital channels are in general not always mixed unitary. Moreover, even the problem of deciding if a given channel is mixed unitary is NP-hard \cite{nphard}, which further limits the usefulness of this approach if we want to implement general channels. Other approaches to reducing the cost of channel implementation include finding a low Kraus rank approximation of a channel \cite{lancien-winter} and performing a sequence of measurements with feedback and adaptive operations \cite{lloyd, shen, cnots}.

In this work, we consider an approach which uses classical randomization together with \textit{postselection}, i.e. measuring an output qubit that heralds whether the simulation was successful or not. The simulation succeeds with some probability possibly less than one and upon that event the channel is simulated on a given input exactly, without approximation. This paradigm of simulation has been effectively used for the problem of simulating arbitrary POVMs using projective measurements and classical resources \cite{acin, maciejewski, singal}, culminating in \cite{kotowski-oszmaniec} where the authors prove that every POVM can be simulated with constant (dimension-independent) probability with only one ancillary qubit. Similar notions have recently been used to study implementations of quantum instruments \cite{armin}. Our paper is therefore situated in the line of work that explores fundamental possibilities and limitations of simulating various quantum objects (POVMs, channels, instruments) with constrained quantum resources, but allowing classical randomness and postprocessing. We do not consider specific practical hardware implementations, but rather study the problem through an information- and resource-theoretic lens.

Our main result is an asymptotically tight characterization of the possible tradeoff between success probability and amount of ancilla needed for the class of unital channels. We give a simulation protocol that implements a given unital channel on $n$ qubits with constant (dimension-independent) success probability using a very small (only $\log n$ qubits) ancilla. Our protocol in general provides success probability $\sim \frac{2^l}{n}$ when using $l$ ancilla qubits. We then show that this is best possible by constructing a family of channels, which we call the epsilon-net channels, that cannot be simulated with probability better than $\sim \frac{2^l}{n}$ when using only $l$ ancilla qubits. For random channels and a more general class of \textit{highly noncommutative} channels in fact one additional qubit suffices to give constant success probability. The methods apply also to channels which are only approximately unital. Finally, we bound the limits of this model of simulation by showing that it necessarily fails for strongly non-unital channels and then propose and study extensions of the model using adaptive operations.

\subsection*{Overview of results}

Any quantum channel $\Phi$ can be implemented by augmenting the system with an additional subsystem (called the ancilla or environment), implementing a unitary on the larger system and then tracing out the environment. This is called the Stinespring dilation of the channel and requires ancilla dimension equal to the Kraus rank of $\Phi$, which for a $d$-dimensional system can be as large as $d^2$ \cite{watrous}. For a system of $n$ qubits this may require up to additional $2n$ ancilla qubits. We are interested in vastly reducing the amount of ancilla necessary by allowing classical randomization, i.e. allowing the unitary to be sampled randomly, and allowing some probability of failure.

We consider the following notion of probabilistic simulation of a quantum channel $\Phi$ using classical randomization and postselection. We equip the system with an additional output flag qubit which, when measured, will herald whether the simulation was successful or not. The allowed simulation procedure (we refer to \cref{def:postselection} for a formal definition) consists of the following steps: 
\begin{enumerate}
    \item Sample outcome $\alpha$ from a classical probability distribution $p_{\alpha}$.
    \item Implement a unitary $U_{\alpha}$ on a larger system and then trace out the ancilla, thus implementing a channel $\Phi_{\alpha}$.
    \item Measure the output flag qubit in the computational basis and decide that simulation was successful upon measuring $0$ and failed otherwise.
\end{enumerate}

 The channels $\Phi_\alpha$ are chosen so that, given an input state $\rho$, if measuring the flag qubit gives $0$, the resulting state on the remaining (non-ancilla) qubits is $\Phi(\rho)$. In other words, conditionally on simulation being successful we have implemented the channel $\Phi$ exactly (with no error). 
 
Our main results give a complete characterization of the asymptotic tradeoff between ancilla dimension and success probability for unital channels. We first show a lower bound on success probability, formally stated as \cref{th:prob-succ-groups-k}.

\begin{result}\label{res:1}
    Every unital channel $\Phi$ in dimension $d$ can be simulated with ancilla dimension $k$ with success probability at least
    \begin{equation}
        q_{\mathrm{succ}} = \Omega\left(\min\left\{1, \frac{k}{\log d}\right\}\right).
    \end{equation}
    Equivalently, any unital channel $\Phi$ on $n$ qubits can be simulated with $l$ ancilla qubits with success probability at least $q=\Omega(\frac{2^l}{n})$. In particular, every unital channel on $n$ qubits can be simulated with probability $q=\Omega(1)$ (dimension-independent) with only $O(\log n)$ ancilla qubits.
\end{result}
Postselection allows us to repeat simulation rounds until the simulation succeeds, assuming we can obtain a new copy of the input state in each round. In this way, using a very small ancilla (only $\log n$ additional qubits, compared to $2n$ qubits necessary for Stinespring dilation) we can simulate any unital channel with success probability arbitrarily close to one by repetition.

We next show that the above result is in general best possible, as there exists a family of unital channels whose simulation success probability saturates the stated lower bound. We have thus completely solved the problem of asymptotically characterizing the maximal possible success probability when constrained by given ancilla dimension. This is formally proved in \cref{th:epsilon-net-bound}.

\begin{result}\label{res:2}
    For each $d$ there exists a unital channel $\phinet$ in dimension $d$, which we call the \emph{epsilon-net channel}, that cannot be simulated with ancilla dimension $k$ with success probability higher than $q=O(\frac{k}{\log d})$.
\end{result}

Finally, we show a broad class of channels for which the bound from \cref{res:1} can be improved to $q=\Omega(1)$, i.e. constant, dimension-independent success probability, with just one additional qubit. This is possible if the Choi matrix $J(\Phi)$ has small operator norm. We call such channels \textit{highly noncommutative} for reasons explained in \cref{sec:highly-noncommutative}. In particular, a random channel with high probability will belong to this class. A precise statement of this result is given in \cref{th:unital-improved}.

\begin{result}\label{res:3}
Any unital channel whose Choi matrix $J(\Phi)$ has sufficiently small operator norm can be simulated with success probability $q=\Omega(1)$ with just one ancillary qubit.
\end{result}

Thus, highly noncommutative channels can be simulated cheaply. In contrast, the epsilon-net channel from \cref{res:2}, which is the hardest to simulate, is a commutative (Schur) channel, i.e. its Kraus operators commute. This is in agreement with the intuition that concentration inequalities from random matrix theory, upon which our protocol is based, work best when applied to highly noncommutative matrices and worst for commutative ones.

The results above are limited to channels which are unital or close to unital. We formally prove in \cref{sec:non-unital} that strongly non-unital channels cannot be simulated in the above sense with small ancilla and high success probability. Thus, we delineate the power of this simulation model. We then propose extending the model with adaptive operations and provide a simulation protocol with small ancilla for measure-and-prepare channels.

\subsection*{Outline of proof}

Let us now briefly sketch how the results are proved. The simulation protocol from \cref{res:1} requires partitioning the set of Kraus operators of $\Phi$ into small subsets, whose size depends on the allowed ancilla dimension, and then randomly choosing one of the subsets and implementing a corresponding channel. This partitioning scheme is inspired by a similar scheme used in \cite{kotowski-oszmaniec} to simulate POVMs. We describe it in detail in \cref{sec:protocol}. 

To ensure high success probability, the protocol requires a partition of Kraus operators which is ``flat'', i.e. all components have small operator norm. In the case of one ancilla qubit this simply means that each Kraus operator $K_i$ should have small norm. To obtain such a partition we exploit non-uniqueness of Kraus representation -- for a given Kraus representation, we use a Haar random isometry to produce another Kraus representation of the same channel which with high probability satisfies the ``flatness'' condition. To analyze the random isometry mathematically we use matrix concentration inequalities of Tropp \cite{tropp-book}, in particular matrix martingale inequalities \cite{tropp-tail}. This is done in detail in \cref{sec:logd-simulation}.

To obtain \cref{res:3} we need stronger matrix inequalities, namely a refined version of noncommutative Khintchine inequalities recently established by Bandeira, Boedihardjo and van Handel \cite{vanHandel}. We note that this variant of noncommutative Khintchine inequalities has recently been successfully used in quantum information to analyze quantum expanders \cite{lancien, lancien2}, Haar random codes \cite{wright}, single-qubit tomography \cite{pauli} and random fermionic states \cite{fermionic} -- we hope that our paper will further popularize their use in the quantum information context. We describe this approach in \cref{sec:noncomm}. Finally, the epsilon-net channel from \cref{res:2} is introduced and analyzed in \cref{sec:epsilon-net}.

\subsection*{Open problems}


We propose several open problems and possible extensions of our work. The first is to extend randomized simulation methods beyond unital channels. Of particular interest is to study the additional power given by measurements and adaptive operations. Specifically we formulate \cref{q:adaptive} -- can \emph{any} quantum channel $\Phi$ be simulated with constant success probability and one ancilla qubit by performing a constant number of adaptive rounds? It is also interesting to study simulations of specific classes of channels arising in practice, e.g. Lindbladian evolutions $\Phi_t = e^{t \mathcal{L}}$ used in dissipative computations.

The second problem is to simulate further quantum operations, such as instruments or superchannels. It would also be highly relevant to consider the notion of approximate simulation, where we only require the resulting channel to be $\varepsilon$-close in the diamond distance to the target channel. For small $\varepsilon$ this notion is more operationally relevant than exact simulation. Lastly, our construction depends on a randomized choice of good Kraus representations and it would be desirable to make this step deterministic and explicit, e.g. by finding a good Kraus representation algorithmically.


\subsection*{Acknowledgments}
We thank Michał Oszmaniec for extensive discussions about the project and also Chris Jones, Grzegorz Rajchel-Mieldzioć and Zbigniew Puchała for useful discussions. We thank Radosław Adamczak for bringing to our attention the noncommutative Khintchine inequalities. The authors acknowledge support from National Science Center, Poland within the QuantERA III Programme (No 2023/05/Y/ST2/00140 acronym Tuquan). The C4QEC project is carried out within the IRAP of the Foundation for Polish Science co-financed by the European Union.

\subsection*{AI use disclosure}
ChatGPT Pro 5.4 and an internal OpenAI model were used in discovering and analyzing the epsilon-net channel -- their output was checked and validated by the authors and all the writing is fully human.

\section{Setup: quantum channels and randomized simulation}


\subsection{Background on quantum channels}

For a finite-dimensional complex Hilbert space $\mathcal{H}$ let $B(\mathcal{H})$ denote the space of all linear operators on $\mathcal{H}$. By $\id_m$ we will denote the identity map on $\C^m$. For an operator $A \in B(\mathcal{H})$ we will denote its operator norm (equal to the largest singular value) by $\norm{A}_{\infty}$. Note that $\norm{AA^{\dagger}}_{\infty} = \norm{A^{\dagger}A}_{\infty} = \norm{A}^2_{\infty}$. We shall denote the positive semidefinite order on matrices by $\preccurlyeq$.

We first recall some basic notions about quantum channels -- for a comprehensive treatment, including proofs of the propositions mentioned, see e.g. the classic textbook \cite{watrous}.  A linear map $\Phi: B(\mathcal{H}_1) \to B(\mathcal{H}_2)$ is positive if $\Phi(X) \succcurlyeq 0$ whenever $X  \succcurlyeq 0 $. A map $\Phi$ is completely positive if $\Phi \otimes \id_m$ is positive for every integer $m\geq 1$.

\begin{definition}\label{def:channel-watrous}
    A quantum channel $\Phi: B(\C^{d}) \to B(\C^{d'})$ is a linear map satisfying any of the equivalent conditions:
    \begin{enumerate}
        \item $\Phi$ is trace-preserving and completely positive.
        \item There exist operators $K_i: \C^{d}\to\C^{d'}, i=1,\dots,n$, with $n \leq d d'$, called the {\normalfont Kraus representation} of $\Phi$, satisfying
        \begin{equation}\label{eq:channel-condition}
            \sum_{i=1}^{n}K^{\dagger}_i K_i = \id_d
        \end{equation}
        such that
        \begin{equation}\label{eq:kraus-representation}
        \Phi(X) = \sum_{i=1}^{n} K_i X K^{\dagger}_i.
        \end{equation}
        \item There exists an isometry $U: \C^{d} \to \C^{d'} \otimes \C^n$, with $n \leq d d'$, such that
        \begin{equation}
            \Phi(X) = \tr_E\left( U X U^{\dagger} \right),
        \end{equation}
        where $\tr_E$ denotes the partial trace over the environment $\C^n$. Any such $U$ is called a {\normalfont Stinespring dilation} of $\Phi$.
        \item The (unnormalized) {\normalfont Choi matrix} $J(\Phi)$, which is an operator on $\C^{d}\otimes\C^{d'}$ defined as
        \begin{equation}\label{eq:choi}
            J(\Phi) := (\id_{d} \otimes \Phi)(\ketbra{\Psi}{\Psi}),
        \end{equation}
        where $\ket{\Psi} = \sum_{i=1}^{d}\ket{i}\otimes\ket{i}$ is the (unnormalized) maximally entangled state, is positive semidefinite and satisfies $\tr_{2}J(\Phi) = \id_d$.
        \end{enumerate}
        The minimal dimension $n$ of the environment of a Stinespring dilation is equal to the minimal number of Kraus operators in the Kraus representation \eqref{eq:kraus-representation}, called the {\normalfont Kraus rank} of $\Phi$. Both numbers are equal to the rank of the Choi matrix $J(\Phi)$.
\end{definition}

We will denote the set of all quantum channels from $B(\C^{d})$ to $B(\C^{d'})$ by $\channels(\C^{d},\C^{d'})$. Unless stated otherwise we assume the input and output dimensions are the same, $d=d'$, and we denote the corresponding set of channels simply by $\channels(\C^d)$.

An important class of channels are unital channels, which are the quantum analogue of doubly stochastic matrices.

\begin{definition}\label{def:unital}
    A channel $\Phi \in \channels(\C^d)$ is called {\normalfont unital} if it preserves the identity matrix, i.e. $\Phi(\id_d)=\id_d$. Equivalently, its Kraus representation satisfies
    \begin{equation}\label{eq:channel-unital}
        \sum_{i=1}^{n}K_i K^{\dagger}_i = \id_d.
    \end{equation}
\end{definition}

We now recall some basic properties of channels. The first one, important for our later applications, is that the Kraus representation of a channel is not unique -- any two Kraus representations of the same channel $\Phi$ are related by an isometry (\cite[Corollary 2.23]{watrous}).

\begin{proposition}\label{prop:kraus-unitary}
    Let $\{K_i\}_{i=1}^{n}, \{L_i\}_{i=1}^{n}$ be two Kraus representations of the same channel $\Phi$. Then there exists a unitary $U: \C^n \to \C^n$ such that for each $i=1,\ldots,n$
    \begin{equation}
        L_i = \sum_{j=1}^{n}U_{ij}K_j.
    \end{equation}
    The case of two Kraus representations of unequal sizes can be handled by padding one of them with zero operators.
\end{proposition}

A common way of realizing the Stinespring dilation is by choosing an isometry $U: \C^{d}\otimes \C^{n} \to \C^{d'} \otimes \C^n$ and writing
\begin{equation}
    \Phi(\rho) := \tr_{E}\left( U \rho \otimes \ketbra{0}{0} U^{\dagger} \right).
\end{equation}
In such setting the Kraus operators $K_i$ can be read off from $U$, in a fixed basis, as the subsequent $d' \times d$ blocks corresponding to the first $d$ columns of the matrix $U$ of size $nd' \times nd$.

A channel $\Phi$ is called \emph{extremal} if it cannot be written as a nontrivial convex combination of two channels, i.e. if $\Phi = p \Phi_1 + (1-p)\Phi_2$ for $p \in (0,1)$ implies $\Phi_1=\Phi_2$. Since the set $\channels(\C^{d},\C^{d'})$ is a compact convex subset of the space of all linear maps from $B(\C^{d_1})$ to $B(\C^{d_2})$, the Krein-Milman theorem states any channel $\Phi$ can be decomposed as a finite convex combination of extremal channels. 




We are interested in implementing a given channel $\Phi$ using unitary operations, allowing the use of an ancillary space of some fixed dimension $k$. The equivalence from \cref{def:channel-watrous} gives one obvious way of implementing a given channel -- simply use its Stinespring dilation. In general this method may require ancilla of dimension $k=d^2$, as the maximal Kraus rank of a channel is $d^2$ and there exist channels which achieve this bound. We will be interested in reducing this ancilla dimension to a much smaller size.

To this end, consider the use of classical randomization -- for a given set of channels $\{\Phi_i\}_{i=1}^{m}$ and a probability distribution $\{p_i\}_{i=1}^{m}$ we implement the channel $\Phi_i$ with probability $p_i$. The channel implemented by such a procedure is equal to the convex combination of the channels $\Phi_i$:
\begin{equation}
    \Phi = \sum_{i=1}^{m}p_i \Phi_i.
\end{equation}

An important and natural class of channels obtained in this way are mixed unitary channels. A channel $\Phi$ is called \emph{mixed unitary} if $\Phi$ is a convex combination of unitary channels, i.e. for some probability distribution $\{p_i\}_{i=1}^{m}$ and unitary operators $\{U_i\}_{i=1}^{m}$ we have
    \begin{equation}
        \Phi(X) = \sum_{i=1}^{m} p_i U_iXU^{\dagger}_i.
    \end{equation}

A mixed unitary channel can be implemented with unitary operations and classical randomization with no ancilla. In general any channel can be implemented with unitary operations and classical randomization using ancilla of dimension $d$ -- first decompose the channel into a convex combination of extremal channels and then simulate the extremal channels using Stinespring dilation. Any extremal channel has Kraus rank at most $d$, which follows from a linear algebraic criterion for extremality, see \cite[Theorem 2.31]{watrous}. Thus the Stinespring dilation of each extremal channel in the decomposition can be implemented using ancilla of dimension at most $d$.

    

    In general, it would be desirable to be able to decompose a given channel into a convex combination of channels with low Kraus rank, as this implies simulability using randomization with small ancilla. Note, however, that already deciding if a channel is a mixed unitary channel or not is NP-hard \cite{nphard}. We conjecture that the same is true for the more general problem of deciding if a channel is in the convex hull of channels of low Kraus rank.

\subsection{The simulation protocol}\label{sec:protocol}

To further reduce the ancilla dimension, let us relax the requirement that the channel be implemented perfectly and allow simulation protocols which only succeed with some probability possibly smaller than one. In this setting it is natural to allow postselection using an additional output qubit, which can be measured to determine whether the simulation was successful. The class of allowed simulation protocols is captured by the following definition:

\begin{definition}\label{def:postselection}
    We say that a channel $\Phi \in \channels(\C^d)$ can be {\normalfont simulated with classical randomization and postselection} with ancilla dimension $k$ and success probability $q \in [0,1]$ if there exists a channel $\Phi' \in \channels(\C^d, \C^d\otimes \C^2)$ such that for any input state $\rho$
\begin{equation}\label{eq:postselection}
        \Phi'(\rho) = q \cdot \Phi(\rho) \otimes \ketbra{0}{0} + (1-q)\Phi^{\junk}(\rho)\otimes \ketbra{1}  {1},
    \end{equation}
    where $\Phi' = \sum_{\alpha} p_{\alpha} \Phi_{\alpha}$, each $\Phi_{\alpha}$ can be implemented using a unitary with ancilla of dimension at most $k$ and the $\Phi^{\junk}$ channel is arbitrary.
\end{definition}


The definition corresponds to a simulation protocol which first samples $\alpha$ from the classical distribution $p_\alpha$ and then implements $\Phi_\alpha$ using its Stinespring dilation -- by our assumption this can be done using ancilla of dimension $k$. Then a measurement of the output flag qubit in the computational basis $\{\ket{0},\ket{1}\}$ is made. We obtain $0$ with probability $q$ and in this case the post-measurement state is $ \Phi(\rho) \otimes \ketbra{0}{0}$, so upon discarding the flag qubit we have simulated the action of $\Phi$ on $\rho$ exactly. On the other hand, if we obtain $1$, we decide that the simulation has failed. In the sequel, whenever we speak of simulating a channel we will mean simulating with classical randomization and postselection in the sense of \cref{def:postselection}.




We now describe the main simulation protocol used in our paper, which is based on the Kraus representation of a channel. The protocol will simulate a given channel using ancilla dimension $k+1$ for some fixed $k$ and will employ a suitable partition of Kraus operators into groups of size at most $k$.

\begin{proposition}\label{prop:simulation-protocol}

Let $\Phi \in \channels(\C^d)$ be a channel with Kraus representation $\{K_i\}_{i=1}^{n}$. Fix a partition of $\{1,\dots,n\}$ into $r$ disjoint subsets $\{X_{\alpha}\}_{\alpha=1}^{r}$, where each subset is of size at most $k$. For each $\alpha=1,\dots, r$ define the channel $\Phi^{(\alpha)}$ with one additional output qubit
\begin{equation}\label{eq:phi_alpha}
    \Phi^{(\alpha)}(\rho) := \lambda_\alpha \sum_{i\in X_\alpha} K_i \rho K_i^\dagger \otimes \ketbra{0}{0} + J^{(\alpha)}\rho (J^{(\alpha)})^{\dagger} \otimes \ketbra{1}{1},
\end{equation}
where
\begin{equation}
J^{(\alpha)} := \sqrt{\id_d - \lambda_{\alpha}\sum_{i\in X_{\alpha}}K^{\dagger}_iK_i}
\end{equation}
and $\lambda_{\alpha} := \norm{\sum_{i\in X_{\alpha}}K^{\dagger}_iK_i}^{-1}_{\infty}$. Define the probability distribution
\begin{equation}
    p_{\alpha} := \frac{\lambda^{-1}_{\alpha}}{\sum_{\beta=1}^{r} \lambda^{-1}_{\beta}}, \quad \alpha = 1, \ldots, r
\end{equation}
and consider the simulation protocol with the following steps:
\begin{enumerate}
    \item Sample $\alpha$ according to the distribution $p_{\alpha}$.
    \item Implement the channel $\Phi^{(\alpha)}$ using its Stinespring dilation.
    \item Measure the flag qubit in the computational basis and then postselect on the measurement result being $0$.
\end{enumerate} 
The protocol simulates the channel $\Phi$ using ancilla dimension at most $k+1$ with success probability
\begin{equation}\label{eq:qsucc}
    q_{\mathrm{succ}} := \frac{1}{\sum_{\alpha=1}^{r} \norm{\sum_{i\in X_{\alpha}}K^{\dagger}_iK_i}_{\infty}}.
\end{equation}
\end{proposition}

\begin{proof}
First, note that $\lambda_{\alpha}$ is the maximal $\lambda$ such that $\id_d \succcurlyeq \lambda\sum_{i\in X_{\alpha}}K^{\dagger}_iK_i$, so $J^{(\alpha)}$ is well-defined. It is also easy to see that $\Phi^{(\alpha)}$ is indeed a valid channel with Kraus rank at most $k+1$ -- the Kraus operators of $\Phi^{(\alpha)}$ are given by $\widetilde{K_i}\ket{\psi} = \sqrt{\lambda_{\alpha}}(K_i\ket{\psi})\otimes\ket{0}$ and $\widetilde{J^{(\alpha)}}\ket{\psi} = (J^{(\alpha)}\ket{\psi})\otimes\ket{1}$.

We readily verify that
\begin{align}
    &\sum_{\alpha=1}^{r}p_{\alpha}\Phi^{(\alpha)} = \sum_{\alpha=1}^{r}\frac{\lambda^{-1}_{\alpha}}{\sum_{\beta=1}^{r} \lambda^{-1}_{\beta}} \cdot \left( \lambda_\alpha \sum_{i\in X_\alpha}K_i \rho K_i^\dagger \otimes \ketbra{0}{0} + J^{(\alpha)}\rho (J^{(\alpha)})^{\dagger} \otimes \ketbra{1}{1} \right) = \\
   &  \left( \sum_{\alpha=1}^{r} \norm{\sum_{i\in X_{\alpha}}K^{\dagger}_iK_i}_{\infty}\right)^{-1} \left( \sum_{\alpha=1}^{r}\sum_{i\in X_\alpha} K_i \rho K_i^\dagger\right) \otimes \ketbra{0}{0} +(1-q_{\mathrm{succ}})\Phi^{\junk} \otimes\ketbra{1}{1} =\\
   &q_{\mathrm{succ}} \Phi\otimes\ketbra{0}{0} + (1-q_{\mathrm{succ}})\Phi^{\junk} \otimes\ketbra{1}{1},
\end{align}
where the exact form of $\Phi^{\junk}$ is unimportant. Thus, the above procedure simulates the channel $\Phi$ with postselection achieving success probability $q_{\mathrm{succ}}$. Since we assumed that the size of each subset $X_{\alpha}$ is at most $k$, by \eqref{eq:phi_alpha} each of the channels $\Phi^{(\alpha)}$ has Kraus rank at most $k+1$. As such, each of them can be implemented using Stinespring dilation with ancilla dimension at most $k+1$. Therefore, the above protocol simulates $\Phi$ with postselection using ancilla dimension at most $k+1$ and with success probability $q_{\mathrm{succ}}$ as defined in \eqref{eq:qsucc}.
\end{proof} 


The protocol above has an inherent tradeoff between the size $k$ of the subsets in the partition, which determines the necessary ancilla dimension, and the success probability $q_{\mathrm{succ}}$. One particular case which will be of interest in the sequel is simulation with partitions consisting of one element subsets, in which case each $\Phi^{(\alpha)}$ has Kraus rank two and thus requires only a single ancilla qubit to implement. \cref{prop:simulation-protocol} thus implies the following

\begin{corollary}\label{cor:main-corollary}
    Any channel $\Phi$ with Kraus representation $\{ K_i  \}_{i=1}^{n}$ can be simulated with a single additional qubit with success probability
    \begin{equation}
        q_{\mathrm{succ}} =  \frac{1}{\sum_{i=1}^{n}\norm{K^{\dagger}_iK_i}_{\infty}}.
    \end{equation}
    In particular, since for each $i$ we have $\norm{K^{\dagger}_iK_i}_{\infty} \leq 1$, the success probability is always at least $\frac{1}{n}$.
\end{corollary}


The efficacy of the protocol depends on the Kraus representation $\{K_i\}_{i=1}^{n}$. The following example illustrates that for a given channel the issue of finding a good Kraus representation, i.e. one for which the sum $\sum_{i=1}^{n}\norm{K^{\dagger}_iK_i}_{\infty}$ is small, can be nontrivial. Let $\Phi$ be the dephasing channel, i.e.
\begin{equation}
    \Phi(\rho) = \sum_{i=1}^{d}\ketbra{i}{i}\rho\ketbra{i}{i}.
\end{equation}
Each of the Kraus operators $K_i=\ketbra{i}{i}$ is a projector and therefore $\norm{K^{\dagger}_iK_i}_{\infty}=1$, which implies that $\sum_{i=1}^{d}\norm{K^{\dagger}_iK_i}_{\infty}=d$ and the success probability given by the protocol is merely $\frac{1}{d}$.  However, it is well-known (\cite[Proposition 4.6]{watrous}) that the dephasing channel is in fact a mixed unitary channel and as such can be simulated perfectly with no ancilla. Note that being mixed unitary implies that there exists a different Kraus representation $\{L_i\}_{i=1}^{n}$, with $L_i = \sqrt{p_i} U_i$ and $U_i$ unitary, such that $\sum_{i=1}^{n}\norm{L^{\dagger}_iL_i}_{\infty}=1$. This naturally invites the following problem.

\begin{problem}
     Given a channel $\Phi$, find the Kraus representation $\{K_i\}_{i=1}^{n}$ and the partition for which the sum $ \sum_{\alpha=1}^{r} \norm{\sum_{i\in X_{\alpha}}K^{\dagger}_iK_i}_{\infty}$ is minimal, given that $\vert X_{\alpha}\vert \leq k$ for each $\alpha$. 
\end{problem}

Since by \cref{prop:kraus-unitary} any two Kraus representations are related by a unitary operator, the above problem can be equivalently posed as follows -- for a given Kraus representation $\{K_i\}_{i=1}^{n}$ find a unitary $U$ such that for new Kraus operators $\{L_i\}_{i=1}^{m}$ given by
    \begin{equation}\label{eq:new-kraus-representation}
        L_i = \sum_{j=1}^{n}U_{ij}K_j
    \end{equation}
    the sum $\sum_{\alpha=1}^{r} \norm{\sum_{i\in X_{\alpha}}L^{\dagger}_iL_i}_{\infty}$ is small.

\section{Simulation for general unital channels}\label{sec:logd-simulation}

Let us specialize to the class of unital channels (recall \cref{def:unital}). This section is devoted to applying the simulation protocol from \cref{prop:simulation-protocol} to prove one of our main results:

\begin{theorem}\label{th:prob-succ-groups-k}
    Let $\Phi$ be a unital channel in dimension $d$ and let $k = O(\log d)$. Then there exists a Kraus representation $\{L_i\}_{i=1}^{n}$ of $\Phi$ such that the simulation protocol from \cref{prop:simulation-protocol} with ancilla dimension $k+1$ achieves success probability at least $q_{\mathrm{succ}} = \Omega\left( \frac{k}{\log d} \right)$.
\end{theorem}

Let us state for further reference two immediate corollaries of the theorem.

\begin{corollary}\label{cor:unital-1-log-d}
    Any unital channel in dimension $d$ can be simulated with success probability $\Omega\left( \frac{1}{\log d} \right)$ by using only one additional qubit.
\end{corollary}

\begin{corollary}\label{cor:unital-log-n-qubits}
    Any unital channel acting on a system of $n$ qubits can be simulated with constant (dimension-independent) success probability by using only $O(\log n)$ ancillary qubits.
\end{corollary}

The Kraus representation $\{L_i\}_{i=1}^{n}$ from \cref{th:prob-succ-groups-k} will be obtained by transforming an arbitrary Kraus representation of $\Phi$ by a random unitary. Namely, let $\{K_i\}_{i=1}^{n}$ be any Kraus representation of $\Phi$ and let $U = (u_{ij})_{i,j=1}^{n}$ be an $n$-dimensional Haar random unitary. We define
\begin{equation}\label{eq:li-def}
    L_i = \sum\limits_{j=1}^{n} u_{ij} K_j, \quad i = 1, \ldots, n.
\end{equation}

Recalling the discussion from the previous section, to prove \cref{th:prob-succ-groups-k} we have to ensure that any group of size $k$, such as for example $\sum_{i=1}^{k} L_{i}^{\dagger} L_{i}$, will have small operator norm. The key ingredient in proving that is the following tail bound:

\begin{proposition}\label{prop:k-groups-tail-bound}
Let $k \leq \alpha \log d$, where $\alpha > 0$ is a constant independent of $d$. Then for $L_i$ defined as in \eqref{eq:li-def} we have for all $d$ large enough
\begin{equation}
    \Pp\left( \norm{ \sum\limits_{i=1}^{k} L^{\dagger}_{i} L_i }_{\infty} > C \frac{\log d}{n} \right) \leq \frac{1}{d^{2+\delta}},
\end{equation}
where $C > 0$, $\delta > 0$ are constants depending only on $\alpha$.
\end{proposition}

Let us first see how the proposition enables us to prove \cref{th:prob-succ-groups-k}.

\begin{proof}[Proof of \cref{th:prob-succ-groups-k}]
    Let $\{K_i\}_{i=1}^{n}$ be any Kraus representation of $\Phi$ and let $\{L_i\}_{i=1}^{n}$ be as in \eqref{eq:li-def}. Without loss of generality we can assume $n \leq d^2$. By padding the Kraus representation with zero operators if needed we can assume that $n$ is divisible by $k$, so that $r = \frac{n}{k}$ is an integer. Let us consider a partition of $\{1, \ldots, n\}$ into groups $\{X_\alpha\}_{\alpha=1}^{r}$ of size $k$, where $X_\alpha$ is given by $\{ (\alpha - 1) k + 1, \ldots, (\alpha - 1) k + k \}$. Recall that if we use $\{L_i\}_{i=1}^{n}$ as a Kraus representation of $\Phi$, the success probability of the protocol from \cref{prop:simulation-protocol} will be given by
    \begin{equation}
        q_{\mathrm{succ}} = \frac{1}{\sum_{\alpha=1}^{r} \norm{\sum_{i\in X_{\alpha}}L^{\dagger}_i L_i}_{\infty}}.
    \end{equation}
    Thus to obtain success probability $\Omega\left( \frac{k}{\log d} \right)$ it is enough to prove that with positive probability over the choice of $U$ we have an upper bound of the form
    \begin{equation}\label{eq:norm-bound}
    \sum_{\alpha=1}^{r} \norm{\sum_{i\in X_{\alpha}}L^{\dagger}_i L_i}_{\infty} = O\left( \frac{\log d}{k} \right).
    \end{equation}
    Since we have $r$ groups and $r k = n$, by a union bound we get for any $C > 0$
    \begin{equation}\label{eq:groups-k-union-bound}
        \Pp\left( \sum_{\alpha=1}^{r} \norm{\sum_{i\in X_{\alpha}}L^{\dagger}_i L_i}_{\infty} > C \frac{\log d}{k} \right) \leq \sum\limits_{\alpha=1}^{r} \Pp\left( \norm{\sum\limits_{i \in X_\alpha} L^{\dagger}_i L_i}_{\infty} > C \frac{\log d}{n} \right).
    \end{equation}
    Note that under a random choice of $U$ the marginal distribution of each term $\norm{\sum_{i\in X_{\alpha}}L^{\dagger}_i L_i}_{\infty}$ is the same. Thus we have to bound
    \begin{equation}
        r \cdot \Pp\left( \norm{\sum\limits_{i =1}^{k} L^{\dagger}_i L_i}_{\infty} > C \frac{\log d}{n} \right).
    \end{equation}
    Since we assume $k = O(\log d)$, \cref{prop:k-groups-tail-bound} implies that for a suitable choice of $C>0$ and $\delta > 0$ the probability above is bounded by $d^{-(2+\delta)}$ for large enough $d$. As $r =  \frac{n}{k} \leq d^2$, we obtain
    \begin{equation}
        r \cdot \Pp\left( \norm{\sum\limits_{i =1}^{k} L^{\dagger}_i L_i}_{\infty} > C \frac{\log d}{n} \right) \leq d^2 \cdot d^{-(2 + \delta)} = o(1)
    \end{equation}
    as $d \to \infty$. This implies that the bound \eqref{eq:norm-bound}
 holds with high probability over the choice of $U$, which finishes the proof.
 \end{proof}

It remains to prove \cref{prop:k-groups-tail-bound}. The proof will rely on estimating matrix-valued moment generating functions using lemmas of technical nature, whose proofs we give in \cref{sec:appendix}. 

\begin{proof}[Proof of \cref{prop:k-groups-tail-bound}]
We will employ the setup of matrix-valued random processes from \cite{tropp2011}. Let $u_1, \ldots, u_n$ denote the rows of $U$ and let $\FF_i$ be the sigma-algebra generated by $\{u_1, \ldots, u_i\}$. Note that conditionally on $\FF_{i-1}$, the row $u_{i}$ is distributed uniformly on the unit sphere in the (random) subspace $H_{i}$ orthogonal in $\C^n$ to $\Span\{u_1, \ldots, u_{i-1}\}$. Let $\E_{i}$ denote the conditional expectation with respect to $\FF_i$.

We will apply \cite[Theorem 2.3]{tropp2011} with $\mathbf{V}_i = \id_d$ and $w=k$ -- in this case it says that if for each $i=1, \ldots, k$ and some $\theta>0$ we have a bound of the form
\begin{equation}\label{eq:cond-mgf}
    \E_{i-1} e^{\theta L_{i}^{\dagger}L_{i}} \preccurlyeq e^{g(\theta) \id_d},
\end{equation}
then the following tail bound holds
\begin{equation}\label{eq:tail-bound-theta}
    \Pp\left( \norm{ \sum\limits_{i=1}^{k} L^{\dagger}_{i} L_i }_{\infty} \geq t \right) \leq d e^{- \theta t + g(\theta)k}.
\end{equation}

Let $V^{(i)}: \C^n \to \C^n$ be any isometry mapping $H_i$ to the subspace $\{(v, 0, \ldots, 0) \, | \, v \in \C^{n-i+1}\}$. Note that the conditional distribution of $z = V^{(i)} u_i$ is uniform on the unit sphere in this subspace, which for simplicity we identify with $\C^{n-i+1}$. Thus conditionally on $\FF_{i-1}$ the operator
\begin{equation}
    L_i = \sum\limits_{j=1}^{n} u_{ij} K_j
\end{equation}
has the same distribution as 
\begin{equation}
    A^{(i)} = \sum\limits_{j=1}^{n-i+1} z_j A_{j}^{(i)},
\end{equation} where $A^{(i)}_{j} = \sum\limits_{\ell=1}^{n} \overline{V}^{(i)}_{j \ell} K_\ell$. Note that the operators $\{A_{j}^{(i)}\}_{j=1}^{n}$ are related to $\{K_{j}\}_{j=1}^{n}$ by a partial isometry, so they satisfy the assumptions of \cref{lem:words} and \cref{lem:mgf-bound}. Hence \cref{lem:mgf-bound} applied with $L = A^{(i)}$ and $D = n-i+1$ implies that for any $\theta \in [0, n-i+1)$ we have
\begin{equation}
    \E_{i-1} \, e^{\theta L_{i}^{\dagger} L_i} \preccurlyeq \frac{1}{1 - \frac{\theta}{n-i+1}} \id_d \preccurlyeq \frac{1}{1 - \frac{\theta}{n-k+1}} \id_d.
\end{equation}
Together with \eqref{eq:tail-bound-theta} this leads to
\begin{equation}\label{eq:theta-bound-with-k}
    \Pp\left( \norm{ \sum\limits_{i=1}^{k} L^{\dagger}_{i} L_i }_{\infty} \geq t \right) \leq d e^{- \theta t} \left( \frac{1}{1 - \frac{\theta}{n-k+1}} \right)^k.
\end{equation}
Let us now take $t = C \frac{\log d}{n}$, $\theta = \varepsilon (n-k+1)$ for constants $C > 0$, $\varepsilon \in (0,1)$ to be adjusted later. The right hand side of \eqref{eq:theta-bound-with-k} is at most
\begin{equation}
    d \exp\left\{- C\varepsilon(n-k+1) \frac{\log d}{n} + k \log \frac{1}{1-\varepsilon}\right\} \leq
    d \exp \left\{-\log d \left( C\varepsilon \frac{n - \alpha \log d + 1}{n}  - \alpha \log \frac{1}{1-\varepsilon}\right)\right\}
\end{equation}
We can assume $n \geq d$ (otherwise split any Kraus operator $K_i$ into appropriately many rescaled copies), so we can always choose $\varepsilon$ small enough and $C$ large enough, depending on $\alpha$ only, so that for all $d$ large enough the coefficient of the logarithm is strictly greater than $3$, which proves the proposition. 
\end{proof}

\section{Simulation with one qubit for highly noncommutative channels}\label{sec:noncomm}

The results of previous sections, in particular \cref{cor:unital-log-n-qubits}, imply that any unital channel $\Phi$ on $n$ qubits can be simulated with constant success probability using $\log n$ ancilla qubits. This in general cannot be further reduced, which we show in \cref{sec:epsilon-net}. Nevertheless, for many channels one can, in fact, simulate them with constant probability using only one additional qubit. The crucial property is that the Choi matrix $J(\Phi)$ should be ``well spread'', as measured by its operator norm being small. We call such channels \textit{highly noncommutative}, for reasons rooted in random matrix theory and explained in \cref{sec:highly-noncommutative}.

As a warmup, we first prove that \textit{random channels} can be simulated with constant probability using just one additional qubit. We then state and prove the general result, \cref{th:unital-improved}. The essential mathematical tool used is a refinement of {\it noncommutative Khintchine inequalities} recently discovered in \cite{vanHandel}, which provides an improvement over the matrix concentration inequalities used in \cref{sec:logd-simulation}. These inequalities can be applied in our case if the Kraus operators are highly noncommutative as measured  by a certain random matrix parameter which turns out to be equal to the operator norm of $J(\Phi)$.

\subsection{Random channels}\label{sec:random-channels}

We choose the model of a random channel based on the Stinespring dilation using a Haar random unitary. We fix ancilla dimension equal to $k$, choose a Haar random unitary $U$ acting on $\C^d \otimes \C^{k}$ and trace out the ancilla:
\begin{equation}\label{eq:random-channel}
    \Phi_U(\rho) := \tr_{E}\left( U \rho \otimes \ketbra{0}{0} U^{\dagger} \right).
\end{equation}
The resulting channel $\Phi$ admits a Kraus representation $\{ K_i \}_{i=1}^{k}$ where the operator $K_i$ corresponds to a $d \times d$ submatrix of $U$ obtained from the first $d$ columns and rows $\{(i-1)d+1, \dots, id\}$. 

The above model is arguably the most natural model of choosing a quantum channel at random in the sense that it admits many equivalent natural descriptions, see \cite{random-channels}. In particular, for $k=d^2$ the model is equivalent to choosing a channel uniformly at random from the Lebesgue (flat) measure on the convex set of all quantum channels of given dimensionality.

\begin{theorem}
    Consider a random channel $\Phi$ sampled as above. Then with probability at least $1 - k\exp(-\Omega(d))$ we have
    \begin{equation}
        \sum_{i=1}^{k} \norm{K^{\dagger}_i K_i}_{\infty} < C
    \end{equation}
    for some absolute constant $C>0$. In particular, by \cref{cor:main-corollary} the channel can be simulated with one additional qubit with success probability $\Omega(1)$.
\end{theorem}

\begin{proof}
Since $ \norm{K^{\dagger}_i K_i}_{\infty} = \norm{K_i}_{\infty}^2$, it suffices to show that with high probability for each $i=1,\dots,k$ we have $\norm{K_i}_{\infty}^2 \leq \frac{C}{k}$ for some constant $C$. First fix $i$. Since each $K_i$ is a $d \times d$ submatrix of a Haar random  $kd \times kd$ matrix $U$, we can use the bounds on the norm of submatrices of a Haar random matrix $U$ which were shown in \cite{singal} using concentration of measure on the unitary group:
\begin{equation}\label{eq: concentration-Ui}
    \Pp\left( \norm{K_i}_{\infty} > \frac{2c}{\sqrt{k}} \right) \leq \exp\left( -\frac{c^2}{12} d \right)
\end{equation}
    for some constant $c\approx 3.85$. 
    The inequality \eqref{eq: concentration-Ui} follows from equation (C.37) in \cite{singal} by putting $t=A$ for $A$ defined therein. Finally, the claim follows by a union bound over $i$.
\end{proof}

The above proof gives high success probability as long as $k$ is subexponential in $d$, but in fact can be made to work for arbitrary $k$ with a slightly longer argument -- instead of a union bound one has to use directly Lipschitz concentration of measure on the unitary group for the function $F(U) = \sum_{i=1}^{k} \norm{K^{\dagger}_i K_i}_{\infty}$. 

\begin{remark}\label{rm:random-channels}
  Random channels discussed in the above model are asymptotically close to unital as soon as $k \to \infty$. How close such channels are to completely depolarizing depends on how $k$ scales with the dimension $d$. For $d^{1+\varepsilon} \ll k \ll d^2$ the norm $\norm{J(\Phi)}_{\infty}$ is typically close to $\frac{1}{d^{\varepsilon}}$ and at the same time their diamond distance from the completely depolarizing channel is close to $2$, i.e. almost maximal \cite{random-channels2}.
\end{remark}

\subsection{Highly noncommutative channels}\label{sec:highly-noncommutative}




\cref{cor:unital-1-log-d} from Section \ref{sec:logd-simulation} shows that any unital channel in dimension $d$ can be simulated with success probability $\Omega(\frac{1}{\log d})$ with only one ancillary qubit. Here we show that for channels with small norm of their Choi matrix $J(\Phi)$ one can obtain an improved bound and get success probability $\Omega(1)$. Note that for any channel the following inequalities hold:
\begin{equation}
  \frac{1}{d} \leq \norm{J(\Phi)}_{\infty} \leq d.
\end{equation}
This is because $\frac{1}{d}J(\Phi)$ is a normalized state on a space of dimension $d^2$, so its largest eigenvalue is between $\frac{1}{d^2}$ (achieved for the completely depolarizing channel) and $1$ (achieved for unitary channels). We will be interested in channels for which $\norm{J(\Phi)}_{\infty} \ll 1$ and we call such channels {\it highly noncommutative} for reasons discussed in \cref{rm:noncomm-explanation}.


\begin{theorem}\label{th:unital-improved}
Let $\Phi$ be a unital channel in dimension $d$ and let $J(\Phi)$ be its Choi matrix. Assume that
\begin{equation}\label{eq:JPhi_small}
        \norm{J(\Phi)}_{\infty} = O\left( \frac{1}{\log^3 d} \right).
    \end{equation}
Then there exists a Kraus representation $\{L_i\}_{i=1}^{n}$ of $\Phi$ such that the simulation protocol from \cref{prop:simulation-protocol} with one ancilla qubit achieves constant (dimension-independent) success probability.
\end{theorem}

The condition \eqref{eq:JPhi_small} says that the Choi matrix $J(\Phi)$ is ``well spread'', so it can be understood as the channel being in some sense close to completely depolarizing. This is, however, true only qualitatively. Random channels can satisfy \eqref{eq:JPhi_small} and at the same time have large diamond distance to the completely depolarizing channel (see \cref{rm:random-channels}). It is also easy to give explicit examples, see \cref{ex:holevo}. 


The proof of \cref{th:unital-improved} is based on the same approach as for \cref{th:prob-succ-groups-k} -- given a Kraus representation $\{K_i\}_{i=1}^{n}$ of $\Phi$, we define a new Kraus representation $L_i = \sum\limits_{j=1}^{n} u_{ij} K_j$ by applying a random unitary $U = (u_{ij})_{i,j=1}^{n}$. To bound the operator norms of $L_i$, which control the success probability of the simulation protocol, we employ newly developed tools from random matrix theory, namely the refinement of noncommutative Khintchine inequalities from \cite{vanHandel}.

Let us briefly describe the setup from the random matrix theory and its relation to the relevant parameters of the quantum channel to be simulated. Suppose we have a random matrix


\begin{equation}
M = \sum\limits_{j=1}^{n} g_j K_j,
\end{equation}
where $g_j$ are i.i.d. standard normal variables (which up to normalization play the role of the unitary entries $u_{ij}$) and for now $K_j$ can be arbitrary $d \times d$ complex matrices. We would like to obtain an upper bound on the operator norm $\norm{M}_{\infty}$. This will be expressed in terms of the following two parameters:
\begin{align}
    & \sigma(M)^2 = \max\left\{ \left\Vert \sum_{i=1}^{n} K^{\dagger}_i K_i \right\Vert_{\infty}, \left\Vert \sum_{i=1}^{n}K_i K^{\dagger}_i \right\Vert_{\infty} \right\}, \label{eq:sigma}\\
    & v(M)^2 = \left\Vert \Cov(M) \right\Vert_{\infty}, \label{eq:cov-v}
\end{align}
where $\Cov(M)$ denotes the $d^2 \times d^2$ covariance matrix of entries of $M$:
\begin{equation}\label{eq:covariance}
\Cov(M)_{ij;kl} = \E M_{ij} \overline{M_{kl}}, \ \ i,j,k,l = 1, \ldots, d. 
\end{equation}

Let us provide an interpretation of these parameters in the relevant case when $\{ K_i \}_{i=1}^{n}$ form a Kraus representation of a quantum channel $\Phi$. Then we have $\sum_{i=1}^{n} K^{\dagger}_i K_i = \id_d$ by virtue of $\Phi$ being a quantum channel, so the first term under the maximum in \eqref{eq:sigma} is equal to $1$. The second term measures how far the channel is from being unital -- since $\sum_{i=1}^{n} K_i K^{\dagger}_i = d \cdot \Phi\left( \frac{\id_d}{d} \right)$, which always has norm at least $1$, we obtain
\begin{equation}\label{eq:sigma2}
\sigma(M)^2 = d \left\Vert \Phi\left( \frac{\id_d}{d} \right) \right\Vert_{\infty}.
\end{equation}
As for $v(M)^2$, let $J(\Phi)$ be the Choi matrix of $\Phi$ (recall \eqref{eq:choi}). A direct computation on matrix elements shows that $\Cov(M) = S J(\Phi) S$, where $S$ is the swap operator on $\C^d \otimes \C^d$, $S\ket{ij} = \ket{ji}$. It follows that $\norm{\Cov(M)}_{\infty} = \norm{J(\Phi)}_{\infty}$, so that
\begin{equation}\label{eq:vm2}
    v(M)^2 = \norm{J(\Phi)}_{\infty}.
\end{equation}




\begin{remark}\label{rm:noncomm-explanation}
Let us provide some intuition why the parameter $\norm{J(\Phi)}_{\infty}$, or $v(M)$, quantifies noncommutativity of Kraus operators. We easily see that $v(M)^2 = \sup_{\norm{A}_2=1}\sum_i\vert\tr(A^{\dagger}K_i)\vert^2$, which is the maximum variance of the matrix $M$ along any single direction $A$ in the operator space. If the operators $K_i$ commute, they span a low-dimensional subspace of operators and consequently this maximum variance is large, implying $\norm{J(\Phi)}_{\infty} \geq 1$. In contrast, $K_i$ being noncommutative can force this variance to be spread over many directions and thus small. A detailed discussion in Section 1.4 of \cite{vanHandel} provides a direct explanation of how $v(M)^2$ measures the noncommutativity of $K_i$. We note that a similar, but not identical, quantity called the commutation index has recently been used in the study of disordered Hamiltonians \cite{disordered, SYK}.
\end{remark}


The inequalities from \cite{vanHandel} are based on a comparison between $\norm{M}_{\infty}$ and $\norm{M_{\mathrm{free}}}_{\infty}$, where $M_{\mathrm{free}}$ is a certain operator coming from free probability, associated to $M$. For our purposes it will be only important that the following inequalities are satisfied (see Section 2.1 of \cite{vanHandel})
\begin{equation}\label{eq:m-free}
    \sigma(M) \leq \Vert M_{\mathrm{free}} \Vert_{\infty} \leq \left\Vert \sum\limits_{j=1}^{n} K_j K_j^{\dagger} \right\Vert_{\infty}^{1/2} + \left\Vert \sum\limits_{j=1}^{n} K_j^{\dagger} K_j \right\Vert_{\infty}^{1/2}.
\end{equation}

One last parameter used in \cite{vanHandel} to bound $\norm{M}_{\infty}$ is $\sigma_{\ast}(M)$ -- it will be only important that it satisfies the inequality (again, see Section 2.1 of \cite{vanHandel})
\begin{equation}\label{eq:sigma-ast}
    \sigma_{\ast}(M) \leq \min \left\{\sigma(M), v(M)\right\}.
\end{equation}
Finally, let
\begin{equation}\label{eq:v-tilde}
    \tilde{v}(M) = v(M)^{1/2} \sigma(M)^{1/2}.
\end{equation}
We can now state the general form of the tail bound on $\norm{M}_{\infty}$ from \cite{vanHandel}.
\begin{proposition}[\cite{vanHandel}, Corollary 2.2]\label{prop:bbvh}
    In the setup as above we have
    \begin{equation}\label{eq:vH_mfree}
        \Pp\left( \norm{M}_{\infty} > \norm{M_{\mathrm{free}}}_{\infty} + C \tilde{v}(M) (\log d)^{3/4} + C \sigma_{\ast}(M) t \right) \leq e^{-t^2},
    \end{equation}
    for all $t \geq 0$, where $C > 0$ is a universal constant.
\end{proposition}

Applying the proposition to the particular case of unital channels yields the following.

\begin{lemma}\label[lemma]{lem:vanhandel-general}
    Let $\Phi$ be a unital channel with Kraus representation $\{K_j\}_{j=1}^{n}$ and let
\begin{equation}
    M = \sum\limits_{j=1}^{n} g_j K_j,
\end{equation}
where $g_j$ are i.i.d. standard normal variables. Then we have
\begin{equation}\label{eq:ugly-equation}
   \Pp\left( \norm{M}_{\infty} > 2 + C s_{1}(\Phi) + C s_2(\Phi) \right) = O\left( \frac{1}{d^{2 + \delta}} \right),
\end{equation}
where 
\begin{align}
    & s_1(\Phi) = \left(\norm{J(\Phi)}_{\infty}  \log^3 d\right)^{1/4}, \\
    & s_2(\Phi) = \min\left\{ 1, \norm{J(\Phi)}_{\infty}^{1/2} \right\} (\log d)^{1/2}
\end{align}
and $C>0, \delta > 0$ are universal constants.
\end{lemma}

\begin{proof}
    By assumption of unitality we have $\sigma(M)=1$ (recall \eqref{eq:sigma}). Thus by \eqref{eq:sigma-ast} and \eqref{eq:v-tilde} we get $\sigma_{\ast}(M) \leq \min\{1, v(M)\}$ and $\tilde{v}(M) = v(M)^{1/2}$. By \eqref{eq:m-free} we obtain $\norm{M_{\mathrm{free}}}_{\infty} \leq 2$, which accounts for the first term on the right hand side under the probability in \eqref{eq:ugly-equation}. By \eqref{eq:vm2} we have $v(M) = \norm{J(\Phi)}_{\infty}^{1/2}$, so  the second term in the tail bound is $C s_1 (\Phi)$ and to obtain the third term it is enough to take $t = (2+\delta) (\log d)^{1/2}$.
\end{proof}



With this lemma we can now prove \cref{th:unital-improved}.


\begin{proof}[Proof of \cref{th:unital-improved}]
    Assume without loss of generality that $n \leq d^2$. Let $U = (u_{ij})_{i,j=1}^{n}$ be an $n$-dimensional Haar random unitary and define as in \eqref{eq:li-def}
    \begin{equation}
        L_i = \sum\limits_{j=1}^{n} u_{ij} K_j, \quad i = 1, \ldots, n.
    \end{equation}
    Similarly as in the proof of \cref{th:prob-succ-groups-k}, we will show that with high probability over the choice of $U$ we have
    \begin{equation}
         \sum_{i=1}^{n} \norm{L^{\dagger}_i L_i}_{\infty} \leq C \max \left\{ \norm{J(\Phi)}^{1/2}_{\infty} (\log d)^{3/2}, 1 \right\},
    \end{equation}
    for some universal constant $C>0$, which by \cref{cor:main-corollary} will give us the desired lower bound on $q_{\mathrm{succ}}$.

    Let us for conciseness write $m= \max \left\{ \norm{J(\Phi)}^{1/2}_{\infty} (\log d)^{3/2}, 1 \right\}$. By a union bound over $i=1, \ldots, n$ it is enough to show that for each $i$ we have
    \begin{equation}
        \Pp\left( \norm{L^{\dagger}_i L_i}_{\infty} > C \frac{m}{n} \right) = o\left( \frac{1}{n} \right).
    \end{equation}
    Since $\norm{L^{\dagger}_i L_i}_{\infty} = \norm{L_i}_{\infty}^2$, clearly it is enough to bound
    \begin{equation}\label{eq:bound-on-li}
    \Pp\left( \norm{L_i}_{\infty} > C \sqrt{\frac{m}{n}} \right),
    \end{equation}
    with a possibly different value of $C$. Note that the marginal distribution of $L_i$ is the same as the distribution of $L = \sum\limits_{j=1}^{n} z_j K_j$, where $z = (z_1, \ldots, z_n)$ is a uniformly random vector on the unit sphere in $\C^n$. To bound $\norm{L}_{\infty}$, we replace $z$ by a vector of Gaussian variables divided by its norm and then use the fact that the norm of a Gaussian vector is tightly concentrated.
    
    Let thus
    \begin{equation}
        M = \sum\limits_{j=1}^{n} g_j K_j,
    \end{equation}
    where $g = (g_1, \ldots, g_n)$ is a vector of i.i.d. standard complex normal variables. An application of \cref{lem:vanhandel-general} gives 
    \begin{equation}\label{eq:ugly-equation2}
        \Pp\left( \norm{M}_{\infty} > 2 + C \norm{J(\Phi)}_{\infty}^{1/4} (\log d)^{3/4} + C \left( \min\{ 1, \norm{J(\Phi)}_{\infty}^{1/2} \} \right) (\log d)^{1/2}\right) = O\left( \frac{1}{d^{2+\delta}} \right),
    \end{equation}
    for universal constants $C>0$, $\delta >0$ (note that the lemma concerns real normal variables instead of complex ones, but this will only affect constants). Since we assume $\norm{J(\Phi)}_{\infty} = O\left(\frac{1}{(\log d)^3}\right)$, the first, constant term in the sum dominates. Recalling that $m= \max \left\{ \norm{J(\Phi)}^{1/2}_{\infty} (\log d)^{3/2}, 1 \right\}$, we obtain that for some universal constants $C, \delta > 0$ and all sufficiently large $d$
    \begin{equation}\label{eq:bound-on-m}
        \Pp\left( \norm{M}_{\infty} > C \sqrt{m} \right) = O\left( \frac{1}{d^{2+\delta}} \right).
    \end{equation}
    
    It remains to pass from estimating $\norm{M}_{\infty}$ to estimating $\norm{L}_{\infty}$. Note that the random vector $\left( \frac{g_1}{\norm{g}}, \ldots, \frac{g_n}{\norm{g}} \right)$, where $\norm{\cdot}$ is the Euclidean norm, is uniformly distributed on the unit sphere in $\C^n$. Thus $L$ has the same distribution as $\sum\limits_{j=1}^{n} \frac{g_j}{\norm{g}} K_j $ and we can write
    \begin{equation}
        \Pp\left( \norm{L}_{\infty} > C \sqrt{\frac{m}{n}} \right) = \Pp\left( \norm{\sum\limits_{j=1}^{n} g_j K_j }_{\infty} > C \norm{g} \cdot \sqrt{\frac{m}{n}} \right) = \Pp\left( \norm{M }_{\infty} > C \norm{g} \cdot \sqrt{\frac{m}{n}} \right).
    \end{equation}
It is standard that $\norm{g}$ is tightly concentrated around $\sqrt{n}$, so by conditioning on, say, $\left\{ \norm{g} \geq \frac{\sqrt{n}}{2} \right\}$ we can estimate
\begin{align}
     \Pp\left( \left\Vert M \right\Vert_{\infty} \geq \Vert g \Vert \cdot C \sqrt{\frac{m}{n}} \right) & \leq
\Pp \left( \left\Vert M \right\Vert_{\infty} \geq \frac{\sqrt{n}}{2} \cdot C \sqrt{\frac{m}{n}}\right) + \mathbb{P}\left( \Vert g \Vert < \frac{\sqrt{n}}{2} 
\right) \nonumber \\
& = \Pp \left( \left\Vert M \right\Vert_{\infty} \geq \frac{C}{2} \sqrt{m} \right) + \Pp\left( \Vert g \Vert < \frac{\sqrt{n}}{2} \right).
\end{align}
To bound the first term we use \eqref{eq:bound-on-m}, after possibly adjusting constants, and for the second term we use classical concentration estimates on $\norm{g}$. Altogether this proves \eqref{eq:bound-on-li} and by the previous remarks finishes the proof of the theorem.
    
\end{proof}

\begin{remark}
    For the sake of simplicity we have stated the bound in \cref{th:unital-improved} only in the case of unital channels with sufficiently small norm $\norm{J(\Phi)}_{\infty}$. However, it is straightforward to generalize the bound, for example to non-unital channels, by using \cref{prop:bbvh} with the corresponding values of parameters $\sigma_{\ast}(M)$, $\tilde{v}(M)$. For non-unital channels the success probability obtained will, apart from the norm of the Choi matrix, depend on $ d\norm{\Phi\left( \frac{\id_d}{d} \right)}_{\infty}$ (which is equal to $1$ for unital channels).
\end{remark}

\begin{example}\label{ex:holevo}
    For a nontrivial example where \cref{th:unital-improved} provides an advantage over \cref{cor:unital-1-log-d}, consider the antisymmetric Werner-Holevo channel:
    \begin{equation}
        \Phi_{WH}(X) = \frac{\tr(X)\id_d - X^T}{d-1}.
    \end{equation}
    The channel has the Kraus representation $K_{ij} = \frac{1}{\sqrt{2d-2}}\left(\ketbra{i}{j} - \ketbra{j}{i}\right), i,j=1,\dots,d$ and its (unnormalized) Choi state is equal to $J(\Phi_{WH}) = \frac{2}{d-1}\Pi_{\mathrm{asym}}$, where $\Pi_{\mathrm{asym}}$ is the projector onto the antisymmetric subspace of $\C^d\otimes \C^d$. The channel is unital, $\norm{J(\Phi_{WH})}_{\infty} = \frac{2}{d-1}$ and the diamond distance between $\Phi_{WH}$ and the completely depolarizing channel is easily bounded from below by $1$. If $d$ is odd, the channel is not mixed unitary (\cite[Exercise 4.2]{watrous}). \cref{th:unital-improved} implies that $\Phi_{WH}$ can be simulated with one additional qubit with constant success probability, whereas \cref{cor:unital-1-log-d} for general unital channels would give only success probability $\Omega(\frac{1}{\log d})$. 
\end{example}

\section{The epsilon-net channel}\label{sec:epsilon-net}

We will now construct a channel which saturates the bound on success probability from \cref{th:prob-succ-groups-k} from the previous section -- namely, a channel which cannot be simulated using ancilla dimension $k$ with probability larger than $O(\frac{k}{\log d})$. The channel is a commutative (Schur) channel obtained from an $\varepsilon$-net on the unit sphere and to our knowledge is new.

Recall that an $\varepsilon$-net $S$ of the unit sphere in $\C^n$ is a finite set of unit vectors $\{v_i\}_{i=1}^{\ell}$ such that for every unit vector $x$ there exists $i$ such that $\Vert x - v_i \Vert < \varepsilon$. It is standard (see e.g. \cite{vershynin}) that there exists an $\varepsilon$-net on the unit sphere in $\C^n$ of size at most $(1 + \frac{2}{\varepsilon})^{2n}$.

\begin{definition}
Let $\varepsilon, a > 0$ and let $m = \lceil a \log d \rceil$. Let $S$ be an $\varepsilon$-net on the unit sphere of $\C^m$ of size at most $(1 + \frac{2}{\varepsilon})^{2m}$. Assume that $\varepsilon, a$ are such that $\vert S \vert \leq d$ -- it suffices to take e.g. $\varepsilon=1/2$ and $a=1/4$. Let $A$ be the $d\times m$ matrix whose rows are vectors from $S$, repeated arbitrarily if needed so that the number of rows is exactly $d$. We define the {\normalfont  epsilon-net channel} associated to $S$ to be the channel $\phinet$ with Kraus operators $\{K_i\}_{i=1}^{m}$ given by

\begin{equation}\label{eq:diagonal-channel}
K_i := \sum_{j=1}^{d}A_{ji} \ketbra{j}{j}.
\end{equation}
\end{definition}

Note that all Kraus operators of the channel $\phinet$ commute. Such channels are known as commutative or Schur channels and every such channel acts as entrywise multiplication by a correlation matrix $C$, i.e. a positive semidefinite matrix with ones on the diagonal, $\Phi(X) = C \circ X$ (\cite[Theorem 4.19]{watrous}). Any such channel is automatically unital. We collect technical properties of $\phinet$ that will be useful in the proof of the upper bound on the success probability.

\begin{remark}\label{rm:properties-of-C}
    For every channel of the form \eqref{eq:diagonal-channel} the correlation matrix is equal to $C=AA^{\dagger}$. For $\phinet$ the rows of $A$ form an $\varepsilon$-net, so $A$ has full column rank, equal to $m$. It follows that $A^{\dagger}A$ is invertible. The projector $P_A$ onto the column space of $A$, which is equal to the column space of $C$, is equal to
    \begin{equation}\label{eq:pa-projector}
        P_A = A(A^{\dagger}A)^{-1} A^{\dagger}.
    \end{equation}
    The Moore-Penrose inverse of $C$, denoted by $C^{+}$, is a matrix that satisfies $CC^{+}=C^{+}C=P_A$ (i.e. inverts $C$ on its range and is zero otherwise) and can be expressed as
    \begin{equation}\label{eq:pseudoinverse}
        C^{+} = A(A^{\dagger}A)^{-2}A^{\dagger}.
    \end{equation}
\end{remark}

We will prove that the channel $\phinet$ cannot be simulated with ancilla of dimension $k$ with success probability better than $O(\frac{k}{\log d})$ and thus saturates the bounds from \cref{sec:logd-simulation}. To this end, we establish the following lemma, which bounds the maximal success probability of simulation by a class of admissible channels using a linear witness. The bound is essentially analogous to bounds appearing in the context of general resource theories \cite{regula}.

Let $\mathcal{N} \subseteq \channels(\C^d, \C^d\otimes \C^2)$ be any set of channels with an additional output qubit and let $S_{\Phi} \subseteq \channels(\C^d, \C^d\otimes \C^2)$ be the set of channels that simulate $\Phi$ with some $q\in[0,1]$ in the sense of \cref{def:postselection}. For a channel $\Phi' \in S_{\Phi}$ we let $q(\Phi,\Phi')$ denote the maximal success probability of simulation and we let
\begin{equation}
q(\Phi, \mathcal{N}) := \sup_{\Phi' \in \mathcal{N} \cap S_{\Phi}}q(\Phi,\Phi').
\end{equation}
Finally, for a channel $\Phi' \in \channels(\C^d, \C^d\otimes \C^2)$ let $\Phi'\vert_0$ denote the subchannel on $\operators(\C^d)$ (a completely positive trace nonincreasing map) obtained from $\Phi'$ as $\Phi'\vert_0 := (\id_d \otimes \bra{0})\Phi'(\id_d \otimes \ket{0})$.


\begin{lemma}\label{lm:upper-bound}
Let $X \in \operators(\C^d \otimes \C^d)$ be a nonnegative operator such that $\tr(J(\Phi)X) \neq 0$. Then for any $\Phi \in C(\C^d)$ and any set of channels $\mathcal{N} \subseteq \channels(\C^d, \C^d\otimes \C^2)$ we have
\begin{equation}\label{eq:upper-bound}
q(\Phi, \mathcal{N}) = \frac{\sup_{\Phi' \in \mathcal{N}\cap S_{\Phi}}\tr(J(\Phi'\vert_0)X)}{\tr(J(\Phi)X)}.
\end{equation}
Moreover, every $\Phi' \in \mathcal{N}\cap S_{\Phi}$ satisfies $\supp J(\Phi'\vert_0) \subseteq \supp J(\Phi)$.
\end{lemma}

\begin{proof}
From \cref{def:postselection} we directly obtain $\Phi'\vert_0 = q \Phi$, so also $J(\Phi'\vert_0) = q J(\Phi)$. This immediately implies $\mathrm{supp}J(\Phi'\vert_0) \subseteq \mathrm{supp}J(\Phi)$ (and in fact the supports are equal unless $q=0$). Multiplying both sides by $X$ and taking the trace yields the result.
\end{proof}

Armed with \cref{lm:upper-bound} we are now ready to prove the main upper bound.

\begin{theorem}\label{th:epsilon-net-bound}
Any channel that simulates $\phinet$ with ancilla dimension $k$ can obtain success probability at most $O(\frac{k}{\log d})$.
\end{theorem}

\begin{proof}
Let $\mathcal{N}_k$ be the set of convex combinations of channels of Kraus rank at most $k$. Consider any $\Phi' \in \mathcal{N}_k \cap S_{\phinet}$ and write $\Phi' = \sum_{\alpha}p_{\alpha}\Phi_{\alpha}'$, where each $\Phi'_{\alpha}$ has Kraus rank at most $k$ (note that $\Phi'_{\alpha}$ separately do not necessarily lie in $S_{\phinet}$). For any $X$ by convexity $\tr(J(\Phi'\vert_0)X)$ in \cref{lm:upper-bound} is bounded from above by $\max_{\alpha}\tr(J(\Phi_{\alpha}'\vert_0)X)$. By \cref{lm:upper-bound} we have $\supp J(\Phi'\vert_0) \subseteq \supp J(\phinet)$, so also $\supp J(\Phi_{\alpha}'\vert_0) \subseteq \supp J(\phinet)$ for each $\alpha$. Therefore, the supremum in \cref{lm:upper-bound} can be bounded by the supremum over $\Phi'$ such that $\Phi'$ has Kraus rank at most $k$ and $\supp J(\Phi'\vert_0) \subseteq \supp J(\phinet)$.

Let $C$ be the correlation matrix corresponding to the channel $\phinet$. Since $\phinet(X) = C \circ X$, we readily compute its Choi matrix $J(\phinet)$:
\begin{equation}
J(\phinet) = \sum_{i,j=1}^{d}C_{ij}\ketbra{i}{j}\otimes \ketbra{i}{j}.
\end{equation}

Note that $J(\phinet)$ is supported on the diagonal subspace $D = \mathrm{span}\{\ket{ii}, i=1,\dots d\}$. It is convenient to identify $D$ with $\mathbb{C}^d$ via the isometry $U: \mathbb{C}^d \to D, U\ket{i}=\ket{ii}$. We then have $J(\phinet) = UCU^{\dagger}$. By letting $X:=UC^{+}U^{\dagger}$, where $C^{+}$ is the Moore-Penrose inverse of $C$ from \cref{rm:properties-of-C}, we immediately get that
\begin{equation}\label{eq:denominaor}
    \tr(J(\phinet)X) = \tr(CC^{+}) = \tr P_A = m.
\end{equation}

To upper bound $\tr(J(\Phi'\vert_0)X)$, note that $J(\Phi'\vert_0)$ has rank equal to at most that of $J(\Phi')$, which is at most $k$. Let $P_D$ be the projector onto the diagonal subspace $D$. Since $X$ is supported on $D$, we have 
\begin{equation}
\tr(J(\Phi'\vert_0)X) = \tr(P_D J(\Phi'\vert_0)P_D X).    
\end{equation}
Again, by using the isometry $U$ we can identify $P_D J(\Phi'\vert_0) P_D$ with a positive semidefinite matrix $M$ such that 
\begin{equation}
    \tr(P_D J(\Phi'\vert_0)P_DX) = \tr(MC^{+}).
\end{equation} It is easily checked that since $\Phi'\vert_0$ is trace nonincreasing, we have $M_{ii} \leq 1$ for each $i=1,\dots,d$. Moreover, since $\supp J(\Phi'\vert_0) \subseteq \supp J(\phinet)$, the matrix $M$ is supported on the column space of $A$. Because of that, we can write $M = A T A^{\dagger}$ for some positive semidefinite matrix $T$ of dimension $m \times m$ -- explicitly we can use \eqref{eq:pa-projector} to write

\begin{equation}
    T= (A^{\dagger}A)^{-1} A^{\dagger} M A (A^{\dagger}A)^{-1}.
\end{equation} 
Therefore by \eqref{eq:pseudoinverse} we have
\begin{equation}
    \tr(MC^{+}) = \tr(ATA^{\dagger} A(A^{\dagger}A)^{-2}A^{\dagger}) =\tr(T),
\end{equation}
so it remains to bound $\tr(T)$.

Let $v$ be the eigenvector of $T$ corresponding to its maximal eigenvalue $\lambda$. Since the rows of $A$ form an $\varepsilon$-net, there exists a row $r_i = A^{\dagger} e_i$, where $e_i$ denotes the $i$-th standard basis vector, such that $\Vert  v - r_i \Vert < \varepsilon$, which quickly implies that $\vert \langle v , r_i \rangle\vert^2 > (1 - \varepsilon^2/2)^2$. By writing $T = \lambda vv^{\ast} + \dots$ and using $M_{ii} \leq 1$ we arrive at
\begin{equation}
    1 \geq e^{\ast}_i M e_i = r^{\ast}_i T r_i \geq \lambda \vert \langle v, r_i \rangle \vert^2,
\end{equation}
from which it follows that
\begin{equation}
    \lambda \leq \frac{1}{(1 - \frac{\varepsilon^2}{2})^2}.
\end{equation}
By assumption about the Kraus rank, $J(\Phi'\vert_0)$ has rank at most $k$ and this bound translates throughout the argument to $T$, so in the end
\begin{equation}
    \tr(T) \leq \mathrm{rank}(T) \Vert T \Vert_{\infty} \leq \frac{k}{(1 - \frac{\varepsilon^2}{2})^2}.
\end{equation}
By combining this with \eqref{eq:denominaor} using \cref{lm:upper-bound} we arrive at
\begin{equation}
    q(\phinet,\mathcal{N}_k) \leq  \frac{k}{(1 - \frac{\varepsilon^2}{2})^2}\cdot\frac{1}{m} \leq \frac{1}{a(1 - \frac{\varepsilon^2}{2})^2} \cdot \frac{k}{\log d}
\end{equation}
as desired. 

\end{proof}

\section{Non-unital channels}\label{sec:non-unital}

The simulation protocol analyzed in previous sections requires that the channel be unital (or close to unital). We will now show that any possible simulation strategy will fail if the channel is strongly non-unital. We then discuss possible extensions of the simulation model that can handle non-unital channels.

For concreteness, let us focus on the class of \emph{measure-and-prepare} channels, consisting of channels $\Phi$ acting as
\begin{equation}
    \Phi(X) = \sum_{i=1}^{n} \tr(M_i X)\sigma_i,
\end{equation}
where $\{ \sigma_i \}_{i=1}^{n}$ is an arbitrary collection of states and $\{M_i\}_{i=1}^{n}$ is an arbitrary POVM, i.e. a collection of operators satisfying $M_i \succcurlyeq 0$ and
\begin{equation}
    \sum_{i=1}^{n}M_i = \id_d.
\end{equation}
Operationally, a measure-and-prepare channel can be applied by first implementing the measurement corresponding to the POVM $\{M_i\}_{i=1}^{n}$ and then, upon obtaining the measurement result $i$, preparing the state $\sigma_i$. 

For a specific example, let us consider the \emph{erasure channel} $\erasure$ acting as
\begin{equation}
    \erasure(X) = \tr(X)\ketbra{0}{0}.
\end{equation}
We now show the erasure channel cannot be simulated with high probability without using a large ancilla. To this end, we prove more generally that the class of channels simulable with non-negligible probability with small ancilla contains only channels close to being unital, as measured by the operator norm of $\Phi\left( \frac{\id_d}{d} \right)$ being small.

\begin{proposition}
    For any channel $\Phi$, the success probability of simulating it with ancilla of dimension $k$ is at most
    \begin{equation}
                q \leq \frac{k}{d \norm{\Phi\left( \frac{\id_d}{d} \right)}_{\infty}}.
    \end{equation}
\end{proposition}

\begin{proof}
      Suppose that a channel $\Phi'$ simulates $\Phi$ with success probability $q$, which by \eqref{eq:postselection}  means that for any input state $\rho$
    \begin{equation}
        \Phi'(\rho) = q \Phi(\rho) \otimes \ketbra{0}{0}+ (1-q)\Phi^{\junk}(\rho) \otimes \ketbra{1}{1}.
    \end{equation}
    Let us write $\Phi' = \sum_i p_i \Phi_i$, where each $\Phi_i$ has Kraus rank at most $k$. By putting $\rho=\frac{\id_d}{d}$ and using positivity we obtain
    \begin{equation}
        \Phi'\left(\frac{\id_d}{d}\right) = q \Phi\left(\frac{\id_d}{d}\right)\otimes \ketbra{0}{0}+ (1-q)\Phi^{\junk}\left(\frac{\id_d}{d}\right) \otimes \ketbra{1}{1} \geq q \Phi\left(\frac{\id_d}{d}\right) \otimes \ketbra{0}{0}.
    \end{equation}
    Taking operator norms of both sides and using convexity yields
    \begin{equation}
\max_i\left\{ \norm{\Phi_i\left(\frac{\id_d}{d}\right)}_{\infty}\right \} \geq q \norm{\Phi\left( \frac{\id_d}{d} \right)}_{\infty}.    
\end{equation}
    Letting $\{ K_j \}_{j=1}^{k}$ be the Kraus representation of the channel $\Phi_i$ for which the maximum is attained, we get
    \begin{equation}
        \norm{\Phi_i\left(\frac{\id_d}{d}\right)}_{\infty}=
        \frac{1}{d}\norm{\sum_{j=1}^{k} K_j K^{\dagger}_j }_{\infty} \leq \frac{1}{d}\sum_{j=1}^{k}\norm{K_jK^{\dagger}_j}_{\infty} \leq \frac{k}{d},
    \end{equation}
    where to prove the last inequality we used $\norm{K_jK^{\dagger}_j}_{\infty} = \norm{K^{\dagger}_j K_j}_{\infty}$ and $\norm{K^{\dagger}_j K_j}_{\infty} \leq 1$ because of \eqref{eq:channel-condition}. Taken together these inequalities imply that
    \begin{equation}
        q \leq \frac{k}{d \norm{\Phi\left( \frac{\id_d}{d} \right)}_{\infty}}
    \end{equation}
    as desired.
\end{proof}

For the erasure channel we have $\erasure\left( \frac{\id_d}{d} \right) = \ketbra{0}{0}$, so we immediately obtain:

\begin{corollary}
    The erasure channel $\erasure$ cannot be simulated with ancilla of dimension $k$ with probability higher than $\frac{k}{d}$.
\end{corollary}


Thus, in order to simulate non-unital channels with small ancilla and high success probability we need to expand the allowed class of operations. The erasure channel $\erasure$ can be implemented trivially with no ancilla if we allow projective measurements and adaptive unitary operations based on measurement results. It suffices to first perform the projective measurement $\{\ketbra{i}{i}\}_{i=0}^{d-1}$ and upon obtaining the result $i$ perform any unitary operation $U_i$ such that $U_i\ket{i} = \ket{0}$. In fact, using projective measurements with adaptive operations together with known results about simulating POVMs with just one additional qubit we can simulate any measure-and-prepare channel with one additional qubit and constant success probability.

\begin{proposition}\label{prop:MAP-simulation}
    Any measure-and-prepare channel can be simulated with one additional qubit with success probability $\Omega(1)$ if we allow projective measurements and adaptive operations.
\end{proposition}

\begin{proof}
Let $\Phi$ be a measure-and-prepare channel given by
\begin{equation}
    \Phi(X) = \sum_{i=1}^{n} \tr(M_i X)\sigma_i,
\end{equation}
where $M = \{M_i\}_{i=1}^{n}$ is a POVM and $\{\sigma_i\}_{i=1}^{n}$ is an arbitrary collection of states. Without loss of generality we can assume that each effect $M_i$ has rank one, i.e., $M_i = \alpha_i \ketbra{\psi_i}{\psi_i}$ for some $\alpha_i > 0$ and vectors $\ket{\psi_i}$. Let
\begin{equation}
    \sigma_i = \sum\limits_{j} \lambda_{ij} \ketbra{\phi_{ij}}{\phi_{ij}},
\end{equation}
where $\lambda_{ij} > 0$ sum up to $1$ and for fixed $i$ vectors $\ket{\phi_{ij}}$ are pairwise orthogonal.

By \cite[Result 1]{kotowski-oszmaniec} the POVM $M$ can be simulated with constant success probability by projective measurements employing only one ancillary qubit. More precisely, there exists a probability distribution $\{p_\beta\}$ and a collection of POVMs $\{N^{(\beta)}\}$ such that 1) $L = \sum_\beta p_\beta N^{(\beta)}$ simulates $M$ with postselection probability at least $q = 1/8$, and 2) each POVM $N^{(\beta)}$ has $d+1$ outcomes (with the outcome $d+1$ interpreted as a flag $\emptyset$ indicating failure of the simulation) and can be implemented by a single projective measurement on $\C^d \otimes \C^2$, i.e., with one ancillary qubit. For the projective measurement implementing $N^{(\beta)}$, we formally let it have $n+1$ outcomes (with the outcome $n+1$ denoting the failure flag), but only $d+1$ of them correspond to nonzero Kraus operators, with the non-failure effects described by $\{ \ketbra{\psi_i^{(\beta)}}{\psi_i^{(\beta)}} \}$ for $i$ in some subset of $\{1, \ldots, n\}$.

This yields the following simulation protocol for $M$:
\begin{enumerate}[1)]
    \item sample $\beta$ from the probability distribution $\{p_\beta\}$,
    \item perform the projective measurement implementing the POVM $N^{(\beta)}$, which requires one ancillary qubit,
    \item upon obtaining result $i=1, \ldots, n$ sample $j$ from the probability distribution $\{\lambda_{ij}\}_j$ and perform any unitary operation $U_{ij}^{(\beta)}$ satisfying $U_{ij}^{(\beta)} \ket{\psi_i^{(\beta)}} = \ket{\phi_{ij}}$; if the measurement result is $n+1$, declare failure.
\end{enumerate}
Since $\sum_\beta p_\beta N^{(\beta)}$ simulates $M$ with postselection probability $q$, we have $\sum_\beta p_\beta N^{(\beta)}_{i} = q M_i$ for $i=1,\ldots,n$, which guarantees that, conditioned on the success of the simulation, we obtain the correct distribution of measurement results and hence the correct output states $\sigma_i$.
\end{proof}

It is known that projective measurements and adaptive operations together with just one ancillary qubit are powerful enough to implement {\it any} channel exactly with $\log d$ rounds. This was first shown in \cite{lloyd} without specific claims on the number of rounds and later \cite{andersson, shen} clarified that $\log d$ rounds are sufficient. Their methods do not make use of classical randomization, so it is natural to ask whether classical randomization can reduce the number of measurement rounds required:

\begin{question}\label{q:adaptive}
    Consider the model of simulation with postselection that allows: 1) unitary operations with ancilla; 2) classical randomization; 3) projective measurements; 4) adaptive unitary operations based on measurement results. Can \textit{every} channel $\Phi$ be simulated in this way with success probability $\Omega(1)$ using one additional qubit and $O(1)$ rounds?
\end{question}

\bibliographystyle{alpha}
\bibliography{bibliography}

\appendix
\crefalias{section}{appendix}

\section{Proofs of technical lemmas}\label{sec:appendix}

\begin{lemma}\label{lem:words}
    Let $A_1, \ldots, A_D$ be operators in dimension $d$ satisfying
    \begin{equation}\label{eq:word-kraus}
    \begin{aligned}
        & \sum\limits_{i=1}^{D} A_{i} A_{i}^{\dagger} \preccurlyeq \id_d, \\
        & \sum\limits_{i=1}^{D} A_{i}^{\dagger} A_{i} \preccurlyeq \id_d. 
    \end{aligned}
    \end{equation}
    For $\ell \geq 1$ let $\pi \in S_\ell$ be any permutation on $\ell$ elements. Then we have
    \begin{equation}\label{eq:word-bound}
    \norm{\sum\limits_{j_1, \ldots, j_{\ell}=1}^{D} A_{j_{1}}^{\dagger} A_{j_{\pi(1)}} \ldots A_{j_{\ell}}^{\dagger} A_{j_{\pi(\ell)}}}_{\infty} \leq 1.
    \end{equation}
\end{lemma}


\begin{proof}
    This is a special case of a general result, see \cite[Theorem 4.1]{bordenave-collins}. To obtain \eqref{eq:word-bound}, we consider the setup of \cite{bordenave-collins} with $r = 2 \ell$, a partition $S$ of the form $B_i = \{ 2i-1, 2\pi(i) \}$ (where $i=1, \ldots, \ell$ and for simplicity we denote both the permutation and the partition by $\pi$), $m=D$, $k=1$ and operators $X_{i,j}$ of the form $X_{2i-1,j} = A_{j}^{\dagger}$, $X_{2i,j} = A_{j}$ for $i=1, \ldots, \ell$. Then in the notation of \cite{bordenave-collins} we have exactly
    \begin{equation}
        X_\pi = \sum\limits_{\vec{j} \in \intervalint{m}^k} X_{1, \vec{j}} \ldots X_{r, \vec{j}} = \sum\limits_{j_1, \ldots, j_{\ell}=1}^{D} A_{j_{1}}^{\dagger} A_{j_{\pi(1)}} \ldots A_{j_{\ell}}^{\dagger} A_{j_{\pi(\ell)}}
    \end{equation}
        and
        \begin{equation}
            \begin{aligned}
            & Q_{\pi, 2i-1} = \norm{ \sum\limits_{j \in \intervalint{m}} X_{2i-1, j} X_{2i-1, j}^{\ast}}^{1/2}_{\infty} = \norm{ \sum\limits_{j = 1}^{D} A_{j}^{\dagger} A_j}^{1/2}_{\infty}, \\
            & Q_{\pi, 2i} = \norm{ \sum\limits_{j \in \intervalint{m}} X_{2i, j}^{\ast} X_{2i, j}}^{1/2}_{\infty} = \norm{ \sum\limits_{j = 1}^{D} A_{j} A_{j^{\dagger}}}^{1/2}_{\infty}.
            \end{aligned}
        \end{equation}
        Each of the terms $Q_{\pi,i}$ is at most $1$ thanks to the assumption \eqref{eq:word-kraus}, which by \cite[Theorem 4.1]{bordenave-collins} proves the lemma.
\end{proof}

\begin{lemma}\label{lem:mgf-bound}
    Let $A_1, \ldots, A_D$ be as in \cref{lem:words} and let
    \begin{equation}
        L = \sum\limits_{i=1}^{D} z_i A_i,
    \end{equation}
    where $(z_1, \ldots, z_D)$ is distributed uniformly on the unit sphere in $\C^D$. Then for any $\theta \in [0, D)$ we have
    \begin{equation}\label{eq:mgf}
        \E \, e^{\theta L^\dagger L} \preccurlyeq \frac{1}{1 - \frac{\theta}{D}} \id_d.
    \end{equation}
\end{lemma}

\begin{proof}
    We will prove that for any $\ell \geq 1$ we have the following moment bound:
    \begin{equation}\label{eq:lth-moment}
        \E (L^\dagger L)^\ell \preccurlyeq \frac{\ell!}{D(D+1)\ldots (D+\ell - 1)} \id_d.
    \end{equation}
    The bound \eqref{eq:mgf} will follow, since then
    \begin{equation}
        \begin{aligned}
        \E \, e^{\theta L^\dagger L} = \id_d + \sum\limits_{\ell=1}^{\infty} \frac{\theta^\ell}{\ell!} \E (L^\dagger L)^\ell & \preccurlyeq \id_d + \sum\limits_{\ell=1}^{\infty} \frac{\theta^\ell}{\ell!} \cdot \frac{\ell!}{D(D+1)\ldots (D + \ell-1)} \id _d \\
        & \preccurlyeq \id_d + \sum\limits_{\ell=1}^{\infty} \frac{\theta^\ell}{D^\ell} \id _d = \frac{1}{1 - \frac{\theta}{D}} \id_d.
        \end{aligned}
    \end{equation}
    To show \eqref{eq:lth-moment} we consider the expansion
    \begin{equation}
        \E (L^\dagger L)^\ell = \sum\limits_{\substack{i_1, \ldots, i_\ell, \\ j_1, \ldots, j_\ell = 1}}^{D} \E (z_{i_1}  \ldots z_{i_\ell} \overline{z_{j_1}} \ldots \overline{z_{j_\ell}} ) A_{j_1}^\dagger A_{i_1} \ldots A_{j_\ell}^\dagger A_{i_\ell}.
    \end{equation}
    We have the following formula for the mixed moments:
    \begin{equation}\label{eq:z-mixed-moments}
        \E (z_{i_1}  \ldots z_{i_\ell} \overline{z_{j_1}} \ldots \overline{z_{j_\ell}} ) = \frac{1}{{D(D+1)\ldots (D + \ell-1)}} \sum\limits_{\pi \in S_\ell} \prod\limits_{k=1}^{\ell} \delta_{i_k, j_{\pi(k)}}, 
    \end{equation}
    where $S_\ell$ is the symmetric group on $\ell$ elements. This follows from the standard identity (see \cite[Proposition 6]{harrow})
    \begin{equation}\label{eq:symmetric-subspace}
        \E \ketbra{z}{z}^{\otimes \ell} = \frac{1}{{D(D+1)\ldots (D + \ell-1)}} \sum\limits_{\pi \in S_\ell} P_D(\pi), 
    \end{equation}
    where $\ket{z}$ is uniformly random on the unit sphere in $\C^D$ and $P_D(\pi)$ is given by
    \begin{equation}
        P_D(\pi) = \sum\limits_{i_1, \ldots, i_\ell = 1}^{D} \ketbra{i_{\pi^{-1}(1)}, \ldots, i_{\pi^{-1}(\ell)}}{i_1, \ldots, i_\ell}.
    \end{equation}
    Formula \eqref{eq:z-mixed-moments} follows directly by comparing matrix elements of both sides in \eqref{eq:symmetric-subspace}.

    Thus we can rewrite \eqref{eq:lth-moment} as
    \begin{equation}
        \E (L^\dagger L)^\ell = \frac{1}{{D(D+1)\ldots (D + \ell-1)}} \sum\limits_{\pi \in S_\ell} \sum\limits_{j_1, \ldots, j_\ell = 1}^{D} A_{j_1}^\dagger A_{j_{\pi(1)}} \ldots A_{j_\ell}^\dagger A_{j_{\pi(\ell)}}
    \end{equation}
    and the triangle inequality implies
    \begin{equation}
        \norm{\E (L^\dagger L)^\ell}_{\infty} \leq \frac{1}{{D(D+1)\ldots (D + \ell-1)}} \sum\limits_{\pi \in S_\ell} \norm{\sum\limits_{j_1, \ldots, j_\ell = 1}^{D} A_{j_1}^\dagger A_{j_{\pi(1)}} \ldots A_{j_\ell}^\dagger A_{j_{\pi(\ell)}}}_{\infty}.
    \end{equation}
    By the assumption \eqref{eq:word-kraus}, \cref{lem:words} implies that each norm on the right-hand side is at most $1$, which implies
    \begin{equation}
        \norm{\E (L^\dagger L)^\ell}_{\infty} \leq \frac{1}{{D(D+1)\ldots (D + \ell-1)}} \cdot \ell!
    \end{equation}
    and \eqref{eq:lth-moment} is proved.
\end{proof}


\end{document}